\documentclass[final]{siamltex}

\usepackage[colorlinks=true,urlcolor=blue]{hyperref}
\usepackage{amstext}
\usepackage{amssymb}
\usepackage{amsmath}
\usepackage{mathrsfs}
\usepackage{array}
\usepackage{enumerate}
\usepackage{url}
\usepackage{mathtools}
\usepackage{bm}

\usepackage[]{graphicx}               
\usepackage{epstopdf}
\usepackage{color}
\usepackage{appendix}
\usepackage[section]{placeins}
\usepackage[section]{algorithm}

\def\multiset#1#2{\ensuremath{\left(\kern-.3em\left(\genfrac{}{}{0pt}{}{#1}{#2}\right)\kern-.3em\right)}}

\newtheorem{remark}{Remark}

\def\<{{\langle}} 
\def\>{{\rangle}} 

\def\p{{\partial}} 

\def\dd#1{\displaystyle{#1}} 

\def\bfC{{\mathbf{C}}} 
\def\bfv{{\mathbf{v}}} 

\def\bm#1{\mathbf{#1}} 

\def\8{{\infty}} 

\usepackage{verbatim}  

\title{On Curvature Driven Rotational Diffusion of Protein on Membrane Surface
}

\author{Y.~C.~Zhou\thanks{Department of Mathematics,
        Colorado State University, Fort Collins, Colorado, 80523-1874
        ({\tt yzhou@math.colostate.edu}).}}

\begin{document}

\maketitle

\begin{abstract}
Morphological dynamics of bilayer membrane is intrinsically coupled to the translational and orientational localization of membrane
proteins. In this paper we are concerned with the orientational localization of membrane proteins in the absence of protein 
interaction and correlation. Entropic energy depending on the angular distribution function and the curvature energy depending 
on the principal curvature vectors are introduced to assemble an energy functional for the coupled system. 
Application of the Onsager's variational principle gives rise to a generalized Smoluchowskii equation 
governing the temporal and angular variations of the protein orientation. We prove the existence of the stationary solution of the equation 
as fixed points of a continuous nonlinear nonlocal map, and for biologically relevant conditions we obtain the uniqueness of the
solution. To approximate the stationary solution in the Fourier space we construct an efficient numerical method that reduces the 
expansion and relates the coefficients to the modified Bessel functions of the first kind. Existence and uniqueness of the numerical solution 
are justified for biologically relevant conditions.
\end{abstract}



\pagestyle{myheadings}
\thispagestyle{plain}
\markboth{Y. C. ZHOU}{Surface Rotational Diffusion of Proteins}

\section{Introduction}
When proteins are attached to or embedded in a bilayer membrane, specific curvature will be induced in the membrane whose 
magnitude depends on the protein structure, the membrane composition, and the solvation environment, among others\cite{JarschI2016a}. 
While quantitative determination of this dependence using the first principle still remains a grand challenge for biophysicists 
and mathematicians, geometrical characterization is possible by relating the membrane curvature induced by a single protein to the
membrane composition in specific solvation environment and applying this {\it intrinsic curvature} of the protein
to a distribution of proteins in membrane. The {\it intrinsic mean curvature} of the protein is recently defined 
and used along with the intrinsic mean curvature of lipids to define the intrinsic mean curvature of the lipid-protein complex
with varying lipid composition and protein coverage \cite{ZhouY2017b}. When the membrane curvature around a protein is 
different from protein intrinsic curvature, the membrane will deform and the protein will displaced, leading to a
dynamic coupling between membrane morphology and protein localization. The dominating modes of the protein displacement 
are translational diffusion and rotational diffusion.

Here we are concerned with the lateral rotational diffusion of protein, i.e., the rotation of proteins in the tangential plane
of the membrane surface $S$. This is the dominant mode of the rotational motion of proteins in membrane, as the rotation away from the 
normal direction of the membrane usually hindered by larger energy barriers \cite{DasB2015a}, c.f. Fig. \ref{fig:model}(left). 
It is then sufficient to use a two-dimensional vector $\bm{u}$, or equivalently, a polar angle $\theta$, to characterize the 
orientation of proteins on the tangent 
plane of the membrane surface. Assuming the surface distribution of membrane protein is sparse, we can neglect the interaction and 
correlation among proteins to define an angular distribution function $\psi(\theta,t)$ of proteins independently at each point on the 
surface $S$. The probability of finding the angular coordinate of proteins at time $t$ within the angular 
interval $[\theta, \theta+d\theta]$ is $\psi(\theta;t) d \theta$. Corresponding to the orientational change of membrane proteins
we will use the principal curvature vector of the membrane rather than the generally used mean curvature and Gaussian curvature, 
for they can not characterize the morphological change of the membrane due to the rotation of proteins. Intrinsic principal curvature 
vectors of the lipids and proteins are also introduced, which along with the angular distribution of the proteins define the 
intrinsic principal curvature vector of the lipid-protein complex. The curvature energy of the bilayer membrane is defined to be a 
function of the mismatch between the intrinsic principal curvature vector of the lipid-protein complex and the principal curvature 
vector of the membrane. This curvature energy and the entropic energy for the protein orientation constitute the energy functional 
for the angular distribution of proteins modulated by the membrane curvature. The governing equation for the rotational diffusion 
of the protein is derived by applying the Onsager variational principle to this energy functional. The relatively short time scale of 
rotational diffusion compared to the membrane morphological change allows us to seek the stationary solution of
the angular distribution for a specified membrane curvature. This stationary solution appears as the
fixed point of a bounded nonlocal map on a convex compact set of continuous periodic functions, which permits the 
justification of the existence and uniqueness of the fixed point for biological relevant conditions. Numerical 
approximation of the fixed point shows that it can be given by even order of cosine series, and the expansion 
coefficients are related to the modified Bessel functions of the first kind. 

This paper is organized as follows. In Section \ref{sect:model} we defined the energy functional for protein
angular distribution driven by the mismatch of the principal curvature vectors. The governing equation
for the variation of angular distribution is derived as the gradient flow of this energy functional.
The stationary solution of the equation is analyzed in Section \ref{sect:analysis}, where the Schauder fixed point 
theorem is employed to proved the existence. For biological relevant conditions we can further bound the 
Frechet derivative of the fixed point map to obtain the uniqueness. Fourier series approximation of the
fixed point is presented in Section \ref{sect:numerical}, where an efficient numerical procedure is formulated
thanks to the representation of the expansion coefficients as functions of modified Bessel functions of the
first kind. 

\section{Variational Modeling of Protein Orientation Coupled to Membrane Curvature} \label{sect:model}
\subsection{Energy Functional for Rotational Protein on Membrane Surface}
We model the angular distribution $\psi(\theta)$ of the protein following the gradient flow of an energy functional consisting of two components: 
\begin{equation} \label{eqn:eng_tot}
E(\psi) = E_{ent}(\psi) + E_{cur}(\psi).
\end{equation}
The first component, the entropic energy $E_{ent}$, is given by the classical grand potential functional of the one-particle 
distribution \cite{OnsagerL1949a,DoiM2011a}
\begin{equation} \label{eqn:eng_Eent}
E_{ent}(\psi) =  \int_0^{2 \pi} k_BT \ln \left( \frac{\psi}{\psi_{eq}} \right) \psi d \theta,
\end{equation} 
where $\psi_{eq}$ is the equilibrium distribution of $\psi$ when the entropic energy $E(\psi)$ is minimized
in the absence of external potential. The other component shall be the energy due to mismatch between the membrane curvature 
and protein intrinsic curvature. Instead of using mean curvature or Gaussian curvature, nevertheless, we will use the principal curvatures to 
characterize the curvature energy, for the reasons that (i) mean curvature or Gaussian curvature individually can not characterize the 
orientational dependence of the curvature energy of the interacted protein-membrane system; 
(ii) single mean curvature or Gaussian curvature will fail to describe the bending nature of membrane for 
some of the most interesting membrane morphology. For instances, the mean curvature at some position on the neck of a budding vesicle 
can be zero, and the Gaussian curvature of a tubule is also zero. Fig. \ref{fig:model} illustrates the protein embedding in membrane
and the intrinsic principal curvature vector of a curvature generation protein. It is apparent that the mean curvature and
Gaussian curvature will not change as the protein is rotating laterally in membrane. A curvature vector with the associated 
principal directions is therefore necessary for the geometrical characterization of the lipid-protein complex.
\begin{figure}[!ht]
\begin{center}
\includegraphics[height=2.5cm]{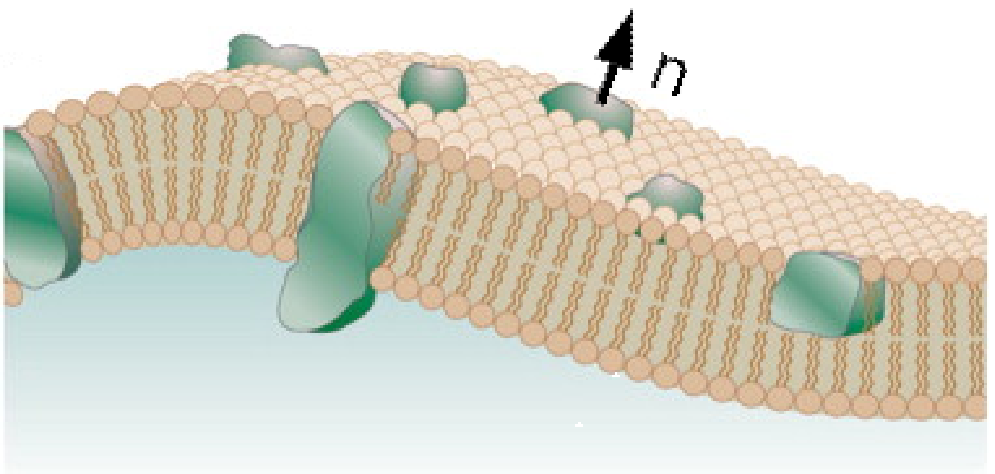}
\includegraphics[height=2.5cm]{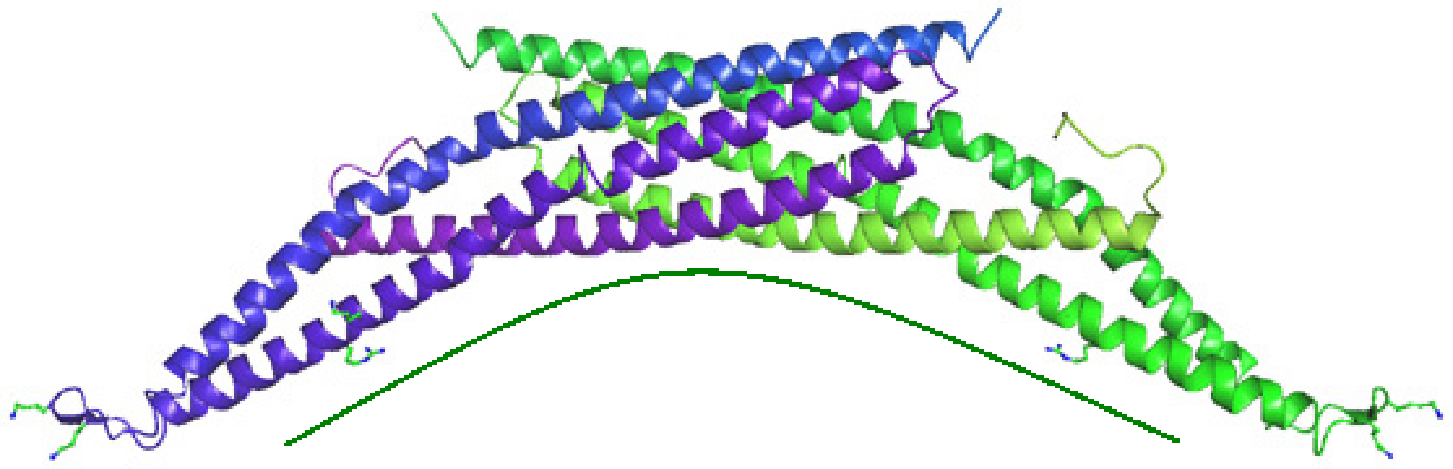}
\caption{Left: Schematic illustration of proteins embedding in a bilayer membrane. Rotational diffusion in the tangent plane with normal direction $n$ 
is energetically more favorable than the rotation away from the direction $n$. Right: A curvature generation protein has an intrinsic curvature vector
$(0,\kappa)$ whose principal direction changes as protein rotates, a nature that can not be characterized by mean curvature or Gaussian curvature only.} 
\label{fig:model}
\end{center}
\end{figure}

The energy $E_{cur}$ we adopt for this second component is the bending energy density of membrane that depends on the angular
distribution of the proteins. This dependence is introduced by considering the intrinsic principal curvature vectors 
$$\bfC^l = (C_1^l,C_2^l), \quad \bfC^p = (C_1^p,C_2^p)$$ 
of the lipids ($l$) and proteins ($p$). The associated principal directions are
$$\bfv^l=(v_1^l, v_2^l), \quad \bfv^p=(v_1^p,v_2^p),$$
where $v_1^l, v_1^p$ are the respective directions of minimum curvatures. 
The principal curvature vector of the membrane and intrinsic principal curvature vector of lipid-protein complex 
are not always aligned, and their mismatch can be sensed and corrected through orientational change of proteins, 
such that
\begin{equation} \label{eqn:eng_Ecur}
E_{cur}(\psi) = \frac{1}{2} K \left( (C_1 - C^0_1(\psi))^2 + (C_2 - C^0_2(\psi))^2 \right),
\end{equation}
is minimized. Here $K$ and $\bfC = (C_1,C_2)$ are the bending modulus and local principal curvature vector of the bilayer membrane,
and $(C_1^0,C_2^0)$ is the local intrinsic principal curvature vector of the lipid-protein complex. In Eq.(\ref{eqn:eng_Ecur}) we use a 
single curvature energy constant $K$ for the two principal curvatures by assuming 
(i) the bilayer membrane structure and its mechanical proprieties are isotropic; and (ii) the asymmetric principal
curvature is the only contribution of the anisotropic nature of the membrane bending energy. 

Despite their overall geometrical configuration (cylindrical or conical), the lipid molecules have much smaller radii than membrane proteins, 
rendering a much larger rotational diffusion coefficients than membrane proteins. We then assume that the principal directions of bilayer membrane are 
always aligned with the intrinsic principal directions of lipids, regardless of their species. Only the membrane proteins might have such 
an orientational mismatch. We further assert that $C_1^p \ne C_2^p$ which holds true for proteins with realistic 
asymmetric molecular structures. It follows that $v_1^p \perp v_2^p$ \cite{Wolfgang_DiffGeo}, 
and thus we can use $(v_1^p,v_2^p)$ as a basis to represent the principal direction $v_1$ of the membrane
\begin{equation}
v_1 = a v_1^p + b v_2^p, 
\end{equation}
where 
\begin{eqnarray}
a &  = & v_1 \cdot v_1^p = \cos(\theta_m - \theta), \\
b &  = & v_1 \cdot v_2^p = \sin(\theta_m - \theta), 
\end{eqnarray}
with $\theta_m,\theta$ being the polar angles respectively corresponding to the principal directions $v_1,v_1^p$ on the local tangent plane of $\mathcal{S}$.
The Euler's theorem asserts that the intrinsic normal curvatures of the proteins in the direction of $v_1$ and $v_2$ are respectively 
given by $ a^2 C_1^p + b^2 C_2^p$ and $b^2 C_1^p + a^2 C_2^p$ \cite{TreatisesDiffGeom}. The local intrinsic curvatures to be established by 
the lipid-protein complex are then
\begin{eqnarray}
C_1^0 & = &  C_1^l f_l + \int_{0}^{2 \pi}(a^2 C_1^p + b^2 C_2^p) f_p d \theta, \label{eqn:c10_final} \\
C_2^0 & = &  C_2^l f_l + \int_{0}^{2 \pi}(b^2 C_1^p + a^2 C_2^p) f_p d \theta. \label{eqn:c20_final}
\end{eqnarray}
Here $0< f_l,f_p<1$ are the surface coverage fractions of lipids and proteins accounting for the number densities $\rho_l,\rho_p$
of lipids and proteins, defined by 
\begin{equation}
f_l = a_l^2 \rho_l/I, \qquad
f_p = a_p^2 \rho_p \psi/I,
\end{equation}
where $a_l,a_p$ are respectively the effective radii of lipids and proteins, and 
\begin{equation} \label{eqn:I}
I = a_l^2 \rho_l + \int_{0}^{2 \pi} a_p^2 \rho_p \psi d \theta = a_l^2 \rho_l + a_p^2 \rho_p = 1, 
\end{equation}
assuming a full coverage of mixed lipids and proteins at an arbitrary position on the membrane surface. 
The curvature energy $E_{cur}$ is then rendered as a function of the protein angular distribution function $\psi$:
\begin{align}
E_{cur}(\psi)  
   = & \frac{1}{2} K \left( d_1 - m \int_{0}^{2 \pi} \alpha_1 \psi d \theta \right)^2 
   +  \frac{1}{2} K \left ( d_2 - m \int_{0}^{2 \pi} \alpha_2 \psi d \theta \right )^2,
\end{align}
where 
\begin{eqnarray*}
d_1 & = & C_1 - C_1^l a_l^2 \rho_l, \\
d_2 & = & C_2 - C_2^l a_l^2 \rho_l, \\
m & = & a_p^2 \rho_p, \\
\alpha_1 & = & a^2 C_1^p + b^2 C_2^p, \\
\alpha_2 & = & b^2 C_2^p + a^2 C_1^p.
\end{eqnarray*}
\subsection{Variational Modeling of Rotational Diffusion}
The time derivative of $E(\psi)$ is
\begin{align}
\dot{E}(\psi)  = & \int_0^{2 \pi} k_BT \left [ \ln \left( \frac{\psi}{\psi_{eq}} \right) + 1 \right] 
 \dot{\psi} d \theta
   - m K \sum_{i=1}^2 \left( d_i - m \int_{0}^{2 \pi} \alpha_i \psi d \theta \right) \int_{0}^{2 \pi} \alpha_i \dot{\psi} d \theta.
\end{align} 
Recall the general conservation equation for the transportation of $\psi$:
\begin{equation} \label{eqn:conserv_1}
\frac{D \psi}{Dt} = \dot{\psi} + \frac{\partial}{\partial \theta} (\dot{\theta} \psi) = 0,
\end{equation}
we will have
\begin{align}
\dot{E}(\psi)  = & -\int_0^{2 \pi} k_BT \left [ \ln \left( \frac{\psi}{\psi_{eq}} \right) + 1 \right ]  \dot{\theta} \frac{\p \psi}{\p \theta} d \theta ~ + \nonumber \\
   &  m K \sum_{i=1}^2 \left( d_i - m \int_{0}^{2 \pi} \alpha_i \psi d \theta \right) \int_{0}^{2 \pi} \alpha_i \dot{\theta} \frac{\p \psi}{\p \theta} d \theta  \nonumber \\
 = & \int_0^{2 \pi} k_BT \dot{\theta} \psi \frac{\partial}{\partial \theta} \left [ \ln \left( \frac{\psi}{\psi_{eq}} \right)  \right ] d \theta - \nonumber \\
  &  m K \sum_{i=1}^2 \left( d_i - m \int_{0}^{2 \pi} \alpha_i \psi d \theta \right) \int_{0}^{2 \pi} \alpha'_i \dot{\theta} \psi d \theta
\end{align} 
after integrating by parts and using the fact that $\psi,\alpha_i,\dot{\theta}$ are periodic, where 
$$\alpha'_i  = \frac{d \alpha_i}{d \theta}.$$
The total rate of change of the Hamiltonian of the lipid-protein complex can now be obtained by assembling
the rate of potential energy dissipation $\dot{E}$ and the rate of kinetic energy dissipation 
\begin{equation}
\Phi(\psi) = \frac{1}{2} \int_0^{2 \pi} \xi_r \dot{\theta}^2 \psi d \theta,
\end{equation}
giving 
\begin{equation}
R = \Phi(\psi) + \dot{E}(\psi). 
\end{equation}
This total dissipation rate will be maximized when 
\begin{equation}
\frac{\delta(\Phi + \dot{E})}{\delta \dot{\theta}} = 0,
\end{equation}
i.e.,
\begin{align}
\xi_r \dot{\theta} \psi + k_B T \psi \frac{\partial}{\partial \theta} \left [ \ln \left( \frac{\psi}{\psi_{eq}} \right) \right ] - 
m K \sum_{i=1}^2  \left( d_i - m \int_{0}^{2 \pi} \alpha_i \psi d \theta \right)  \alpha'_i \psi = 0,
\end{align}
indicating that 
\begin{equation}
\dot{\theta}  = - \frac{k_B T}{\xi_r \psi} \cdot \frac{\partial \psi}{\partial \theta} 
 + \frac{m K}{\xi_r} \sum_{i=1}^2 \left( d_i - m \int_{0}^{2 \pi} \alpha_i \psi d \theta \right)  \alpha'_i. 
\end{equation}
With this form of $\dot{\theta}$ the conservation equation (\ref{eqn:conserv_1}) now reads
\begin{align} 
\dot{\psi} & = \frac{\partial}{\partial \theta} \left [\frac{k_B T}{\xi_r} \cdot \frac{\partial \psi}{\partial \theta} 
- \frac{m K}{\xi_r} \sum_{i=1}^2  \left( d_i - m \int_{0}^{2 \pi} \alpha_i \psi d \theta \right)  \alpha'_i \psi \right ]. 
\label{eqn:conserv_2}
\end{align}
It is indeed an integral-differential equation for $\psi$: 
\begin{equation} \label{eqn:smol_rotdiff}
\frac{\partial \psi}{\partial t} = \frac{\partial}{\partial \theta} \left[  D_r \frac{\partial \psi}{\partial \theta} - D_r
\frac{m K}{k_BT} \sum_{i=1}^2 \left( d_i - m \int_{0}^{2 \pi} \alpha_i \psi d \theta \right)  \alpha'_i \psi 
 \right]. 
\end{equation}
where $D_r = k_BT /\xi_r$ is the rotational diffusion coefficient of proteins. Equation (\ref{eqn:smol_rotdiff}) is a 
Smoluchowski equation (aka Fokker-Planck equation) with a nonlocal potential 
$$-\frac{m K}{k_BT} \sum_{i=1}^2  \left( d_i - m \int_{0}^{2 \pi} \alpha_i \psi d \theta \right)  \alpha'_i$$
that depends on the mismatch of the principle curvatures of lipid membrane and proteins.

\section{Stationary Solution of the Rotational Diffusion} \label{sect:analysis}
The time scale of rotational diffusion of protein is smaller than those of the translational diffusion and of the membrane deformation, 
allowing us to seek the stationary solution of the time-dependent equation (\ref{eqn:smol_rotdiff}) for given membrane curvature 
and local protein number density. The steady state of Equation (\ref{eqn:smol_rotdiff}) reads
\begin{equation} \label{eqn:smol_rotdiff_steady}
\frac{d}{d \theta} \left[ D_r \frac{d \psi}{d \theta} - 
D_r C_E \sum_{i=1}^2  \left( d_i - m \int_{0}^{2 \pi} \alpha_i \psi d \theta \right)  \alpha'_i \psi 
 \right] = 0.
\end{equation}
where $C_E = mK/(k_BT)>0$. By introducing the Slotboom transformation
\begin{eqnarray*}
\tilde{D}_r &  = & D_r C_E e^{\sum_{i=1}^2 \left( d_i - m \int_{0}^{2 \pi} \alpha_i \psi d \theta \right)  \alpha_i }, \\
\tilde{\psi} & = & \psi e^{-C_E \sum_{i=1}^2 \left( d_i - m \int_{0}^{2 \pi} \alpha_i \psi d \theta \right)  \alpha_i },
\end{eqnarray*}
one can convert Equation (\ref{eqn:smol_rotdiff_steady}) to be
\begin{equation} \label{eqn:smol_rotdiff_steady_trans}
\frac{d}{d \theta} \left( \tilde{D}_r  \frac{d \tilde{\psi}}{d \theta} \right) = 0,
\end{equation} 
to which $\tilde{\psi} = N$ is a solution for some real constant $N$. It follows that
\begin{equation} \label{eqn:map_initial}
\psi = N e^{C_E \sum_{i=1}^2 \left( d_i - m \int_{0}^{2 \pi} \alpha_i \psi d \theta \right)  \alpha_i },
\end{equation}
and for $\psi$ to be a probability function in $[0, 2\pi]$ we must choose 
$$ \frac{1}{N} =  
\int_{0}^{2 \pi} e^{C_E \sum_{i=1}^2 \left( d_i - m \int_{0}^{2 \pi} \alpha_i \psi d \theta \right)  \alpha_i } d \theta.$$
Therefore a solution to Equation (\ref{eqn:smol_rotdiff_steady}) is a fixed point of the composite map 
\begin{equation} \label{eqn:map_F}
\mathcal{F}(\psi) = \mathcal{N} \circ \mathcal{E} (\psi),
\end{equation}
with
$$ \mathcal{E}(\psi) =  e^{C_E\sum_{i=1}^2  \left( d_i - m \int_{0}^{2 \pi} \alpha_i \psi d \theta \right)  \alpha_i }$$
and the normalization map 
$$ \mathcal{N}(\psi) = \psi \left( \int_{0}^{2 \pi} \psi d \theta \right)^{-1}.$$

\subsection{Existence of Stationary Solution}
We will establish the existence of a solution $\psi$ to the stationary equation (\ref{eqn:smol_rotdiff_steady}) in the admissible set
\begin{equation} \label{eqn:S_def}
S =\left \{ \psi \in C[0, 2 \pi]:  \psi \ge 0,   \psi(0) = \psi(2 \pi), \int_{0}^{2 \pi} \psi d \theta = 1,  |\psi'(\theta)| \le M \right \},
\end{equation}
with a norm of $C[0,2 \pi]$ being the Banach space
\begin{equation}
\| \psi \|_{\infty} = \max_{\theta \in [0, 2 \pi]} | \psi(\theta)|.
\end{equation}
Here the first three conditions are among the essential properties of $\psi$ as a probability function in $[0, 2\pi]$, and the last condition 
ensures that any sequence of functions in $S$ is equicontinuous. The positive constant $M$ is to be defined below.

\begin{lemma} \label{lemma:S_compact}
The set $S$ defined by (\ref{eqn:S_def}) is a compact, convex set in $C[0,2\pi]$.
\end{lemma}
\begin{proof}
The convexity of the set follows from the quick fact that $t \psi_1 + (1-t) \psi_2$ for any $\psi_1,\psi_2 \in S$ is also in $S$.
The uniform boundedness of $\psi$ and $\psi'$ indicates that any sequence of functions $\{ \psi_n \} \in S$ is bounded and equicontinuous.
It follows from the Arzela-Ascoli theorem that this sequence must have a subsequence $\{ \psi_{nk} \}$ uniformly converging to a function 
$\psi \in C[0,2\pi]$. To show that this function is indeed in $S$, we first assume that $\psi(\theta_0)<0$ at some $\theta_0 \in [0, 2 \pi]$, say
$\psi(\theta_0) =-\epsilon$ where $\epsilon >0$. Then the uniform convergence of the subsequence indicates that there is $K$ such that
$$ | \psi_{nk}(\theta) - \psi(\theta) | < \frac{\epsilon}{2}$$
for all $k>K$ and $\theta \in [0, 2 \pi]$. It follows that 
$$ \psi_{nk}(\theta) < \psi(\theta) + \frac{\epsilon}{2},$$
which violates the fact that $\psi_{nk}(\theta) \ge 0$ at 
$\theta=\theta_0$ considering the assumption that $\psi(\theta_0) =-\epsilon$. Thus $\psi(\theta)$ is bounded from below by $0$.
Following similar arguments one can show that $\psi(\theta)$ is periodic and has an integration of $1$ in $[0,2\pi]$.
Thus $\psi \in S$, and $S$ is compact.
\end{proof}

\begin{lemma} \label{lemma:F_continuous}
$\mathcal{F}(\psi)$ is a continuous map from $S$ to $S$ for $M>0$ sufficiently large.
\end{lemma}
\begin{proof}
Define 
\begin{eqnarray}
C^p_{\max} & = & \max \left( |C_1^p|, |C_2^p| \right), 
\end{eqnarray}
we first estimate 
\begin{eqnarray*}
\sum_{i=1}^2  \left ( d_i - m \int_{0}^{2 \pi} \alpha_i \psi d \theta \right) \alpha_i \le  (|d_1|+|d_2| + 2m C_{\max}^p ) C_{\max}^p := \mathcal{B} 
\end{eqnarray*}
considering $\int_{0}^{2 \pi} \psi d \theta = 1$, which also implies that 
$$ \| \psi \|_{\8} \ge  \frac{1}{2 \pi}$$ 
for $\psi \in S$. Consequently we will have
\begin{equation}
e^{-C_E \mathcal{B}} \le \| \mathcal{E}(\psi) \|_{\8} \le  e^{C_E \mathcal{B}},
\end{equation}
meaning $\mathcal{E}$ is bounded. To further estimate $\mathcal{F}(\psi)$ we notice that
\begin{eqnarray} 
\| \mathcal{F}(\psi) \|_{\8} & = & \| \mathcal{N}\circ \mathcal{E}(\psi) \|_{\8} \nonumber \\
 & \le & \frac{\| \mathcal{E}(\psi) \|_{\8}}{\int_{0}^{2 \pi} \mathcal{E}(\psi) d \theta} \nonumber \\
 & \le & \frac{e^{C_E \mathcal{B}}}{2 \pi e^{-C_E \mathcal{B}}} \nonumber \\
 & = & \frac{e^{2C_E \mathcal{B}}}{2 \pi}   \nonumber \\
 & \le & e^{2C_E \mathcal{B}}  \| \psi \|_{\8} \label{eqn:F_estimate}
\end{eqnarray}
suggesting that the composite map $\mathcal{F}$ is continuous on $C[0, 2\pi]$.

To find the upper bound $M$ of $|\psi'(\theta)|$ we calculate and estimate
\begin{eqnarray}
\left | \frac{d \mathcal{F}(\psi)}{d \theta} \right| & = & \frac{|d\mathcal{E}(\psi)/d \theta|}{\int_{0}^{2 \pi} \mathcal{E}(\psi) d \theta} 
\nonumber \\
 & = & \frac{m K}{k_BT}\sum_{i=1}^2  \left( d_i - m \int_{0}^{2 \pi} \alpha_i \psi d \theta \right)  \alpha'_i \cdot 
\frac{\mathcal{E}(\psi)}{\int_{0}^{2 \pi} \mathcal{E}(\psi) d \theta} \nonumber \\
& \le & \frac{\mathcal{B}}{2 \pi} e^{2\frac{m_K}{k_B T} \mathcal{B}} := M.
\end{eqnarray}

Furthermore, it follows from the definition that $\mathcal{F}$ is nonnegative, and periodic since $a_i$ is periodic. 
In addition $\dd{\int_0^{2 \pi} \mathcal{F}(\psi) d \theta = 1}$ because of the normalization by $\mathcal{N}$. 
Hence $\mathcal{F}$ is indeed a continuous map from $S$ to $S$ with the constant $M$ defined above.
\end{proof}

We are now at the position to prove the existence of the stationary solution. 
\begin{theorem}
Equation (\ref{eqn:smol_rotdiff_steady}) has a solution $\psi \in S$.
\end{theorem}
\begin{proof}
A solution of Equation (\ref{eqn:smol_rotdiff_steady}) is equivalent to a fixed point of the map $\mathcal{F}$, which 
exists following the Schauder fixed point theorem applied to $\mathcal{F}$ on $S$, along with the fact that $S$ 
is a compact subset of $C[0,2 \pi]$ and $\mathcal{F}$ is a continuous map from $S$ to $S$, as justified in Lemma \ref{lemma:S_compact} 
and \ref{lemma:F_continuous}.
\end{proof}

\subsection{Uniqueness of the Stationary Solution}
We are more interested in a constructive fixed point mapping towards the unique solution of 
equation (\ref{eqn:smol_rotdiff_steady}), if exists, for the biologically relevant range of parameters.
To this end we shall explore the contractivity of the map $\mathcal{F}$ in appropriate subset of $C[0,2\pi]$.

We compute the Fr\'echet derivative $\frac{DF}{D\psi}$ of $\mathcal{F}$ as
\begin{eqnarray}
\frac{DF}{D\psi}(w) & = & 
\frac{\mathcal{E}(\psi) \cdot C_E \dd{\sum_{i=1}^2 \left(-m a_i \int_{0}^{2 \pi} a_i w d \theta \right)}}
{\dd{\int_{0}^{2 \pi}\mathcal{E}(\psi) d \theta}} - \nonumber \\
& & 
\frac{\mathcal{E}(\psi) \dd{ \int_{0}^{2 \pi} \mathcal{E}(\psi) \cdot
\left [ C_E \sum_{i=1}^2 (-m a_i) \int_{0}^{2 \pi} a_i w d \theta \right] d \theta}}
{\left( \dd{\int_{0}^{2 \pi}\mathcal{E}(\psi) d \theta} \right )^2}  \nonumber \\
 & = & \mathcal{F} (\psi)  \cdot C_E \sum_{i=1}^2 \left(-m a_i \int_{0}^{2 \pi} a_i w d \theta \right) - \nonumber \\
 &  & \mathcal{F} (\psi) \cdot 
C_E \sum_{i=1}^2 \left[-m \cdot
\frac{\dd{ \int_{0}^{2 \pi} \mathcal{E}(\psi) a_i d \theta}}{\dd{\int_{0}^{2 \pi}\mathcal{E}(\psi) d \theta}} \int_{0}^{2 \pi} a_i w d\theta \right] \nonumber \\
 & = & \mathcal{F} (\psi)  \cdot C_E \sum_{i=1}^2 \left(-m a_i \int_{0}^{2 \pi} a_i w d \theta \right) - \nonumber \\
 &  & \mathcal{F} (\psi) \cdot 
C_E \sum_{i=1}^2 \left( -m \bar{a}_i \int_{0}^{2 \pi} a_i w \theta \right) \nonumber \\
 & = & \mathcal{F} (\psi)  \cdot C_E \sum_{i=1}^2 \left [-m ( a_i - \bar{a}_i ) \int_{0}^{2 \pi} a_i w d \theta  \right]
\end{eqnarray}
for $w \in C[0, 2 \pi]$, where $\bar{a}_i$ arises in the mean value approximation of the integral $\int_{0}^{2 \pi} \mathcal{E}(\psi) a_i d \theta$. 
One can then estimate
\begin{eqnarray}
\left \| \frac{D \mathcal{F}}{D \psi} (w) \right \|_{\infty} & \le & 4 \pi m C_E \cdot \Delta C^p \cdot C^p_{\max} \cdot 
\| \mathcal{F}(\psi) \|_{\8} \| w \|_{\8}, 
\end{eqnarray}
where 
\begin{eqnarray}
\Delta C^p & = & \max \left( C_1^p, C_2^p \right) - \min \left( C_1^p, C_2^p \right).
\end{eqnarray}
It follows that
\begin{eqnarray}
\left \| \frac{D \mathcal{F}}{D \psi} \right \| & \le & 
4 \pi m C_E \cdot \Delta C^p \cdot C^p_{\max} \cdot \| \mathcal{F}(\psi) \|_{\8}  \nonumber \\
& \le & 4 \pi m C_E \cdot \Delta C^p \cdot C^p_{\max} \cdot \frac{1}{2 \pi} e^{2 C_E \mathcal{B}} \nonumber \\
&  = & 2m C_E \cdot \Delta C^p \cdot C^p_{\max} \cdot e^{2 C_E \mathcal{B}}, \label{eqn:DFDphi}
\end{eqnarray}
where we applied the absolute estimate of $\mathcal{F}(\psi)$ obtained right before the inequality (\ref{eqn:F_estimate}).

For the map $\mathcal{F}$ to be contractive we need to have $\| D\mathcal{F}/D \psi \| < 1$, it is sufficient to have
\begin{equation} \label{eqn:contractive_cond_0}
 e^{2 C_E \mathcal{B}} < \frac{1}{2m C_E \cdot \Delta C^p \cdot C^p_{\max}}. 
\end{equation}
in regard to the estimate (\ref{eqn:DFDphi}). Condition (\ref{eqn:contractive_cond_0}) is equivalent to 
\begin{equation} \label{eqn:contractive_cond_1}
 e^{2 \frac{mK}{k_BT} (|d_1| + |d_2| + 2m C^p_{\max}) C^p_{\max} } < \frac{k_BT}{2m^2 K \cdot \Delta C^p \cdot C^p_{\max}}. 
\end{equation}
It is apparent that equation (\ref{eqn:contractive_cond_1}) holds true for any $|d_1|,|d_2|,\Delta C^p$ and $C^p_{\max}$ 
when $m$ is sufficiently small, and thus we shall conclude 
\begin{theorem}
There exists a unique stationary solution of (\ref{eqn:smol_rotdiff_steady}) for any membrane protein when its surface number
density is sufficiently small.
\end{theorem}

\begin{figure}[!ht]
\begin{center}
\includegraphics[width=5cm]{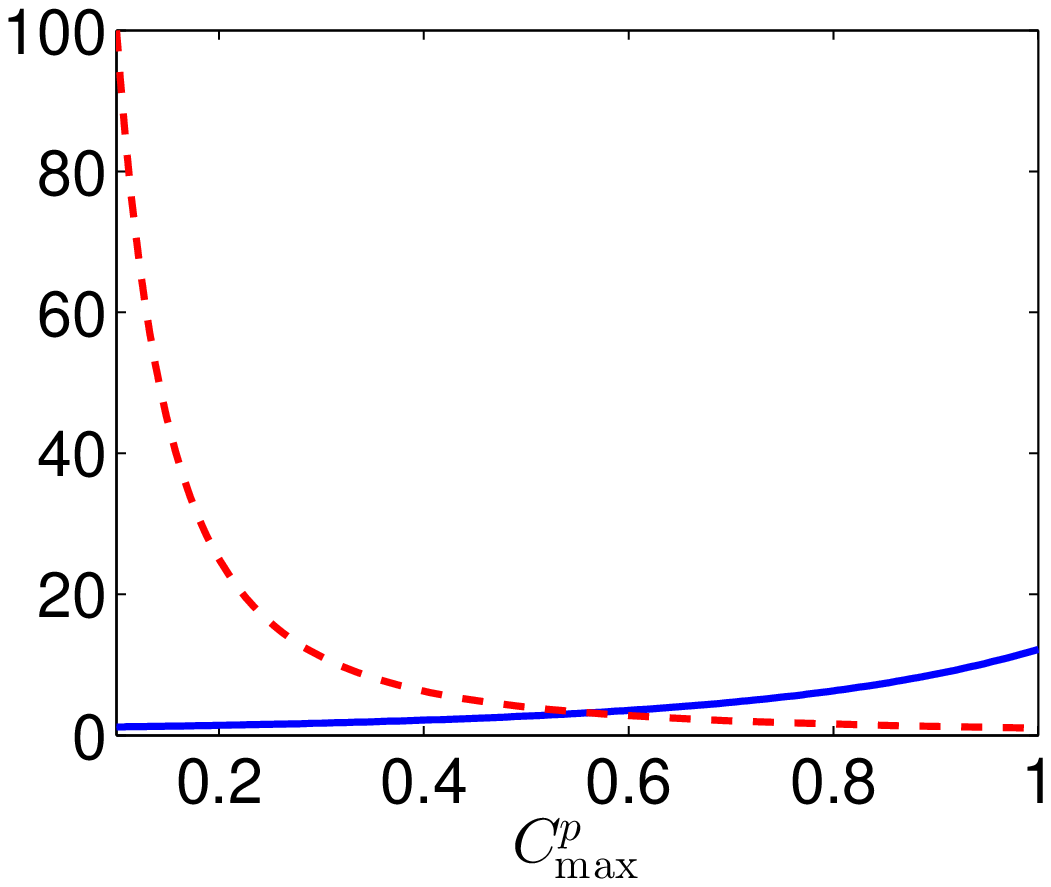} \hspace{1cm}
\includegraphics[width=5cm]{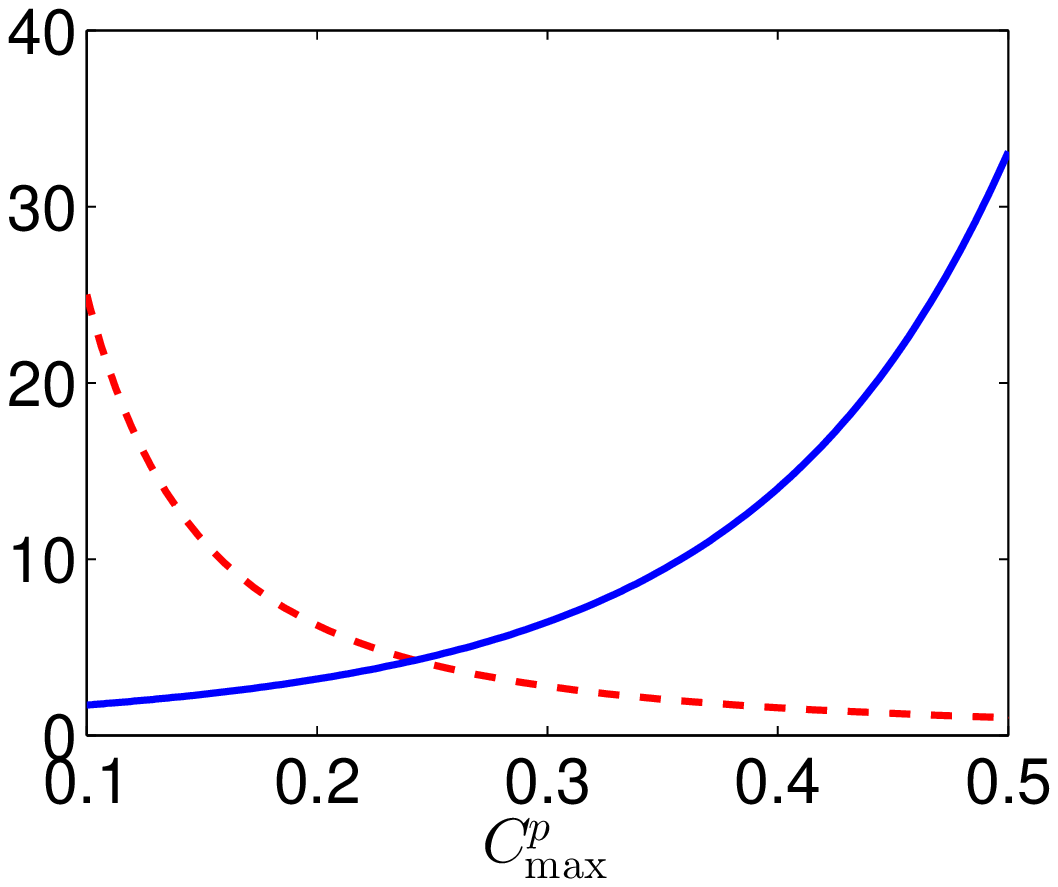}
\caption{Functions on the two sides of inequality (\ref{eqn:contractive_cond_1}) with $K/k_BT=1, |d_1|=|d_2|=1$ and
two values of $m=0.5$ (left chart) and $m=1$ (right chart). The exponential growth function (blue, solid) always intersects with 
the reciprocal of the power function of degree $2$ (red, dash). The value of $C^p_{\max}$ at the intersection
defines the intrinsic curvature of the membrane proteins below which the continuous map $\mathcal{F}(\psi)$ is 
guaranteed to be contractive. In the worst case scenario for which $m=1$, the two function intersect at $C^p_{\max} \approx 0.24$,
which is larger than the known intrinsic curvatures of membrane proteins.}
\label{fig:curv_m}
\end{center}
\end{figure}
The constant $0\le m \le 1$ may be not necessarily small for interested ranges of curvatures and number densities of proteins
on membrane surfaces. Fig (\ref{fig:curv_m}) plots the functions on two sides of inequality (\ref{eqn:contractive_cond_1})
varying with $C^p_{\max}$ for a set of given parameters. Tab (\ref{tab:Cmax_protein}) lists the values 
or ranges of values of these parameters. 
\begin{table}[!ht]
\begin{center}
\begin{tabular}{ll} \hline
Parameter & Value or range of value \\ \hline
$k_B$   & $1.38 \times 10^{-23}$ J$\cdot$K$^{-1}$ \\
$T$     & $200 \sim 400$K \\
$K$     &  $1.5 \sim 5.0 \times 10^{-21}$ J$\cdot$nm$^{2}$ \\
$|d_1|$ & $0 \sim 1$ nm$^{-1}$ \\
$|d_2|$ & $0 \sim 1$ nm$^{-1}$ \\ 
$C^p_{\max}$ & $0 \sim 0.1$ nm$^{-1}$ \\ \hline
\end{tabular}
\caption{Biologically relevant values or ranges of values of the controlling parameters in surface rotational diffusion model
of membrane proteins. Ranges of $|d_1|,|d_2|$ are deduced over all major lipid species including phosphatidylcholines(PC),
phosphatidylethanolamines(PE), and Cholesterols \cite{GrunnerS1985a,KamalM2009a,McmahonT2015a}. Values of $K$ are estimated
based on a membrane compression modulus $K_a = 100 \sim 300 mN/m$, an avergage bilayer thickness $t=4nm$, and the relation
$K = c_m K_a t^2$ where $c_m = $1 nm$^2$. Range of $C^p_{\max}$ is
summarized over five classes of most well understood curvature generation or curvature sensing proteins: BIN/Amphiphysin/Rvs (BAR) domain superfamily 
proteins \cite{QualmannB2011a,MimC2012b,ZimmerbergJ2004a,HabermannB2004a,SimunovicM2015a,PrevostC2015a}; 
The endosomal sorting complex required for transport (ESCRT) \cite{WolletT2009a,LeeI2015a,AdellM2016a,HenneW2013a}; 
group specific antigens (Gag) \cite{WeissE2011a,LeeI2015a}; Matrix-2 proteins (M2) \cite{KovacsF1997,RobertK2013a,RossmanJ2010a}; 
and reticulons (RTN): \cite{HuJ2008a,SanoF2012a,JarschI2016a}.}
\label{tab:Cmax_protein}
\end{center}
\end{table}
For these biologically relevant situations it is sufficient to show that 
\begin{equation} \label{eqn:contractive_cond_2}
 e^{2 \frac{K}{k_BT} (|d_1| + |d_2| + 2C^p_{\max}) C^p_{\max} } < \frac{k_BT}{4K \cdot (C^p_{\max})^2}
\end{equation}
noticing that $\Delta C^p \le 2 C^p_{\max}$. For given membrane curvature, lipid species and 
surface number density of membrane proteins, the only variable in functions on two sides of 
inequality (\ref{eqn:contractive_cond_2}) is $C^p_{\max}$. Solving this inequality for saturated surface membrane coverage with 
$m=1$ gives that 
\begin{equation}
 C^p_{\max} < 0.24.
\end{equation}
This allows us to justify the uniqueness of the fixed point for the map $\mathcal{F}(\psi)$, as stated in the following theorem:
\begin{theorem} \label{thm:unique}
Under biologically relevant conditions there exists a unique stationary solution 
of (\ref{eqn:smol_rotdiff_steady}) for any membrane protein number densities provided that the intrinsic principal curvatures 
of the protein are less than $0.24nm^{-1}$.
\end{theorem}

A trivial stationary solution can be obtained when the protein has no orientational preference, i.e., $C_1^p = C_2^p$.
While in this case these does not exist an associated orthorgonal principal directions, we can choose two arbitrary 
orthorgonal unit vectors as $v_1^p,v_2^p$ to represent $v_1$. It follows that $\alpha_{1,2} = C_1^p = C_2^p$ are 
both constants, and thus the nonlinear map (\ref{eqn:map_F}) will produce 
$$ \mathcal{F}(\psi) = \frac{1}{2 \pi}$$
for any input $\psi$. Therefore the $1/2 \pi$ is the only fixed point of the map, hence the unique stationary solution. This
constant $\psi$ is consistent with the physical nature of the problem for that the protein shall be uniformly orientated if
its has no orientational preference. For solutions of $\psi$ with $C_1^p \ne C_2^p$ in general we will resort to numerical 
methods to be developed below.

\section{Numerical Approximation of the Stationary Solution} \label{sect:numerical}
The nonlinearity of the map (\ref{eqn:map_F}) makes it practically impossible to find its unique fixed point in a close form. For that reason we 
proceed in this section to numerically approximate of the stationary solution $\psi(\theta)$. The periodicity of the fixed point $\psi$ allows us
to do this in the Fourier space, i.e., we will find the numerical solution to Equation (\ref{eqn:smol_rotdiff_steady}) of the form
\begin{equation} \label{eqn:numsol_psi}
\psi_N(\theta) = \sum_{k=0}^{N_{\psi}} C_k \cos(k \theta) + S_k \sin(k \theta)
\end{equation}
for some finite integer $N_{\psi}>0$ and coefficients $\{C_k \}, \{S_k\}$. Continuity of $\psi(\theta)$ indicates that the 
convergence of $\psi_N$ to $\psi = \psi_{\infty}$ is uniform.

We first make an assertion on the reduced form of the solution to Equation (\ref{eqn:smol_rotdiff_steady}):
\begin{proposition}
The solution to Equation (\ref{eqn:smol_rotdiff_steady}) can be given by
\begin{equation} \label{eqn:numsol_psi_cosine}
\psi(\theta) = \sum_{k=0}^{\8} T_k \cos(2k (\theta_m - \theta)). 
\end{equation}
for some coefficients $\{ T_k \}$.
\end{proposition}
\begin{proof}
We examine the solution by checking the fixed point $\psi = \mathcal{F}(\psi)$, which is equivalent to the map defined in (\ref{eqn:map_initial}).
Notice there $N$ is a positive scalar, and thus $\psi(\theta)$ on the left-hand side of the equation must have the same type and number of modes as 
the right-hand side, which can be written as
\begin{eqnarray}
e^{C_E \sum_{i=1}^2 \left( d_i - m \int_{0}^{2 \pi} \alpha_i \psi d \theta \right)  \alpha_i } & = & 
e^{C_E \left( d_1 - m \int_{0}^{2 \pi} \alpha_1 \psi d \theta \right)  (\cos ^2 (\theta_m - \theta) C_1^p + \sin ^2 (\theta_m - \theta) C_2^p)  }  ~ \cdot \nonumber \\
 & & e^{C_E \left( d_2 - m \int_{0}^{2 \pi} \alpha_2 \psi d \theta \right)  (\sin ^2 (\theta_m - \theta) C_1^p + \cos ^2 (\theta_m - \theta) C_2^p)  } \nonumber \\
& = & 
e^{C_E \left( d_1 - m \int_{0}^{2 \pi} \alpha_1 \psi d \theta \right)  (C_2^p + (C_2^p - C_1^p)\cos ^2 (\theta_m - \theta)) } ~ \cdot \nonumber \\
 & & e^{C_E \left( d_2 - m \int_{0}^{2 \pi} \alpha_2 \psi d \theta \right)  (C_1^p + (C_1^p - C_2^p)\cos ^2 (\theta_m - \theta))  } \nonumber \\
& = & e^{p_1} \cdot e^{p_2 \cos^2 (\theta_m - \theta) }, \label{eqn:exp_1}
\end{eqnarray}
where 
\begin{eqnarray}
p_1 & = & C_E \left( d_1 - m \int_0^{2 \pi} \alpha_1 \psi d \theta \right) C_2^p + C_E \left( d_2 - m \int_0^{2 \pi} \alpha_2 \psi d \theta \right) 
C_1^p, \label{eqn:p_1} \\
p_2 & = & C_E \left( d_1 - d_2 - m \int_0^{2 \pi} (\alpha_1 - \alpha_2) \psi d \theta \right) (C_2^p - C_1^p) \label{eqn:p_2}
\end{eqnarray}
are real scalars depending on $\psi$. The second term on the right-hand side of Equation (\ref{eqn:exp_1}) has the Taylor expansion
\begin{equation} \label{eqn:exp_2}
e^{p_2 \cos^2 (\theta_m - \theta) } = \sum_{n=0}^{\8} \frac{p_2^n \cos^{2n}(\theta_m -\theta)}{n!} = \sum_{k=0}^{\8} E_k \cos(2 k (\theta_m - \theta)),
\end{equation}
where we represent each term of $\cos^{2n}(\theta_m - \theta)$ as a summation of $\cos(2k(\theta_m - \theta))$ for non-negative even $k$. 
It follows that the Fourier series of $\psi(\theta)$ will have even cosine modes only.
\end{proof}

We now look at the integrals $\int_{0}^{2 \pi} \alpha_i \psi d \theta$ and $p_1,p_2$ for $\psi$ given by (\ref{eqn:numsol_psi_cosine}). We notice that 
\begin{eqnarray}
\int_{0}^{2 \pi} \alpha_1 \psi d \theta & = & \int_{0}^{2 \pi} \left( \frac{C_1^p + C_2^p}{2} - \frac{C_1^p - C_2^p}{2} \cos(2(\theta_m - \theta)) \right) \cdot \nonumber  \\
 & & \qquad \sum_{k=0}^{\8} T_k \cos(2k (\theta_m - \theta)) d \theta \nonumber \\
 & = & \frac{(C_1^p + C_2^p)T_0}{2} \int_0^{2 \pi} d \theta - \frac{(C_1^p - C_2^p)T_1}{2} \int_0^{2 \pi} \cos^2 (2(\theta_m - \theta)) d \theta \nonumber \\
 & = & \pi (C_1^p + C_2^p)T_0  - \frac{\pi (C_1^p - C_2^p)}{2} T_1, \\
\int_{0}^{2 \pi} \alpha_2 \psi d \theta & = & \pi (C_1^p + C_2^p)T_0  - \frac{\pi (C_2^p - C_1^p)}{2} T_1,
\end{eqnarray}
owing to the orthogonality of $\cos(2k(\theta_1 - \theta))$ to each other on $[0, 2 \pi]$. It follows that $p_1,p_2$ will be linear functions 
of $T_0,T_1$. In particular
\begin{equation} \label{eqn:p2}
p_2 = C_E(d_1- d_2) (C_2^p - C_1^p) - \pi m C_E (C_2^p - C_1^p)^2 T_1
\end{equation}
is a function of $T_1$ only. Conversely, $T_1$ is a a linear function of $p_2$:
\begin{equation} \label{eqn:T1p2}
T_1 = \frac{C_E(d_1- d_2) (C_2^p - C_1^p) - p_2}{\pi m C_E (C_2^p - C_1^p)^2}
\end{equation}
provided that $C_1^p \ne C_2^p$ otherwise a trivial stationary solution has been obtained above.

We are now at the position to present the algebraic relations for computing $\{ T_k \}$. The map (\ref{eqn:map_initial}) now appears
\begin{equation} \label{eqn:exp_3}
e^{\frac{1}{2} p_2 \cos(2(\theta_m - \theta)) } = \sum_{k=0}^{\8} T_k \cos(2k (\theta_m - \theta)) \cdot \int_0^{2 \pi} e^{\frac{1}{2} p_2 \cos(2(\theta_m - \theta)) }d \theta,
\end{equation}
where $e^{p_1 + \frac{1}{2} p_2}$ has been cancelled from both sides. We further denote the integral on the right-hand side of Equation (\ref{eqn:exp_3}) by $I_T$, which 
is a function of $p_2$. Notice that the modified Bessel function of the first kind $I_k(z)$ is given by \cite{Bessel_function}
\begin{equation} \label{eqn:Bessel_function}
I_k(z) = \frac{1}{2 \pi} \int_0^{2 \pi} e^{z \cos (\theta) } \cos(k \theta) d \theta,
\end{equation}
the left-hand side of Equation (\ref{eqn:exp_3}) will has a cosine series 
expansion \footnote{This constitues an alternative justification of the cosine series expansion in Equation (\ref{eqn:exp_2}).}
\begin{equation} \label{eqn:exp_4}
e^{\frac{1}{2} p_2 \cos(2(\theta_m - \theta)) } = \sum_{k=0}^{\8} I_k(p_2/2) \cos(2k (\theta_m - \theta)).
\end{equation}
To collect the coefficents of $\cos(2k(\theta_m - \theta))$ on both sides of Equation (\ref{eqn:exp_3}) we shall get
\begin{eqnarray}
I_0(p_2/2) & = & T_0 \cdot I_T(p_2), \\
I_1(p_2/2) & = & T_1 \cdot I_T(p_2), \label{eqn:T1} \\
I_2(p_2/2) & = & T_2 \cdot I_T(p_2), \\
& \cdots &  \nonumber \\
I_k(p_2/2) & = & T_k \cdot I_T(p_2). 
\end{eqnarray}
Indeed, examining the representation (\ref{eqn:exp_4}) we shall find that 
$$ I_T(p_2) = 2 \pi I_0(p_2/2),$$ 
and thus 
$$ T_0 = \frac{1}{2\pi}.$$
This value turns out to the exact solution of $\psi(\theta)$ for the case that the proteins have no surface orientational preference, i.e., $C_1^p = C_2^p$, 
for which $p_2 = 0$ (c.f. Equation (\ref{eqn:p2})), and thus the Fourier cosine series in Equation (\ref{eqn:exp_4}) has only one single zero order term. 
This conclusion is consistent with the analysis at the end of Section (\ref{sect:analysis}).

To solve for $T_1$ we use the following relation derived from Equations (\ref{eqn:p2}) and (\ref{eqn:T1}):
\begin{equation} \label{eqn:T1_b}
I_1(p_2/2) = 2 \frac{C_E(d_1- d_2) (C_2^p - C_1^p) - p_2}{m C_E (C_2^p - C_1^p)^2} \cdot I_0(p_2),
\end{equation}
which can be written as a fixed point map
\begin{equation} \label{eqn:mapG}
p_2 = \mathcal{G}(p_2) = a_p - \frac{b_p}{2} \cdot \frac{I_1(p_2/2)}{I_0(p_2)},
\end{equation}
with
$$ a_p = C_E(d_1- d_2) (C_2^p - C_1^p), \quad b_p = m C_E (C_2^p - C_1^p)^2.$$
The fixed point $p_2$ can be given by an intersection of the curve $I_1(p/2)/I_0(p)$ and the straight line $2a_p/b_p - 2p/b_p$, 
as illustrated in Fig. \ref{fig:mapG}. Depending on the slope $2/b_p$, there might be one, two, or at most three intersection.
When the slope is sufficiently large, only one interaction can be found, as justified by the following theorem.
\begin{theorem}
The map $\mathcal{G}(p)$ has a unique fixed point under biologically relevant conditions.  
\end{theorem}
\begin{proof}
It is sufficient to show that the mapp $\mathcal{G}$ is contractive, i.e., $|\mathcal{G}'(p)| < 1$, for biologically relevant situations.
Notice that modified Bessel functions of the first kind have the following differentiation relations:
\begin{eqnarray*}
\frac{dI_k(p)}{dp} & = & \frac{k}{p} I_k(p) + I_{k+1}(p), \\ 
\frac{dI_k(p)}{dp} & = & I_{k-1}(p) - \frac{k}{p} I_k(p).
\end{eqnarray*}
It follows that 
\begin{eqnarray*}
\frac{dI_0(p)}{dp} & = & I_{1}(p), \\
\frac{dI_1(p)}{dp} & = & I_0(p) - \frac{1}{p} I_1(p), 
\end{eqnarray*}
and thus
\begin{equation} \label{eqn:dG}
\mathcal{G}'(p) = -\frac{b_p}{2} \cdot \frac{\frac{1}{2} I_0(p/2) I_0(p) - \frac{1}{p} I_1(p/2) I_0(p) - I_1(p/2) I_1(p) }{I^2_0(p)}.
\end{equation}
Denote the second quotient in Equation (\ref{eqn:dG}) by $d\mathcal{G}^I(p)$ whose maximum for $p \in \mathbb{R}$ occurs at $p=0$
where $d\mathcal{G}^I(0) = d\mathcal{G}^I_{\max} =0.25$ (c.f. Fig. \ref{fig:mapG}). To make $|\mathcal{G}'(p)|<1$ it is 
required that 
$$b_p<\frac{2}{d\mathcal{G}^I_{\max}} = 8.$$
On the other hand, under biologically relevant conditions we will have
$$ b_p = m C_E (C_2^p - C_1^p) ^2 \le 4 \frac{m^2K}{k_BT}  (C^p_{\max})^2 \le
 4 \frac{K_{\max}}{k_BT_{\min}}  (C^p_{\max})^2 \approx 0.42$$ 
for $m\le 1, K_{\max} = 5.0 \times 10^{-21}, T_{\min} =200$ and $C^p_{\max} = 0.24$ chosen
in Tab. \ref{tab:Cmax_protein} and Theorem \ref{thm:unique}. This warrants the unique 
fixed point of $\mathcal{G}$. 
\end{proof}
\begin{figure}[!ht]
\begin{center}
\includegraphics[height=4cm]{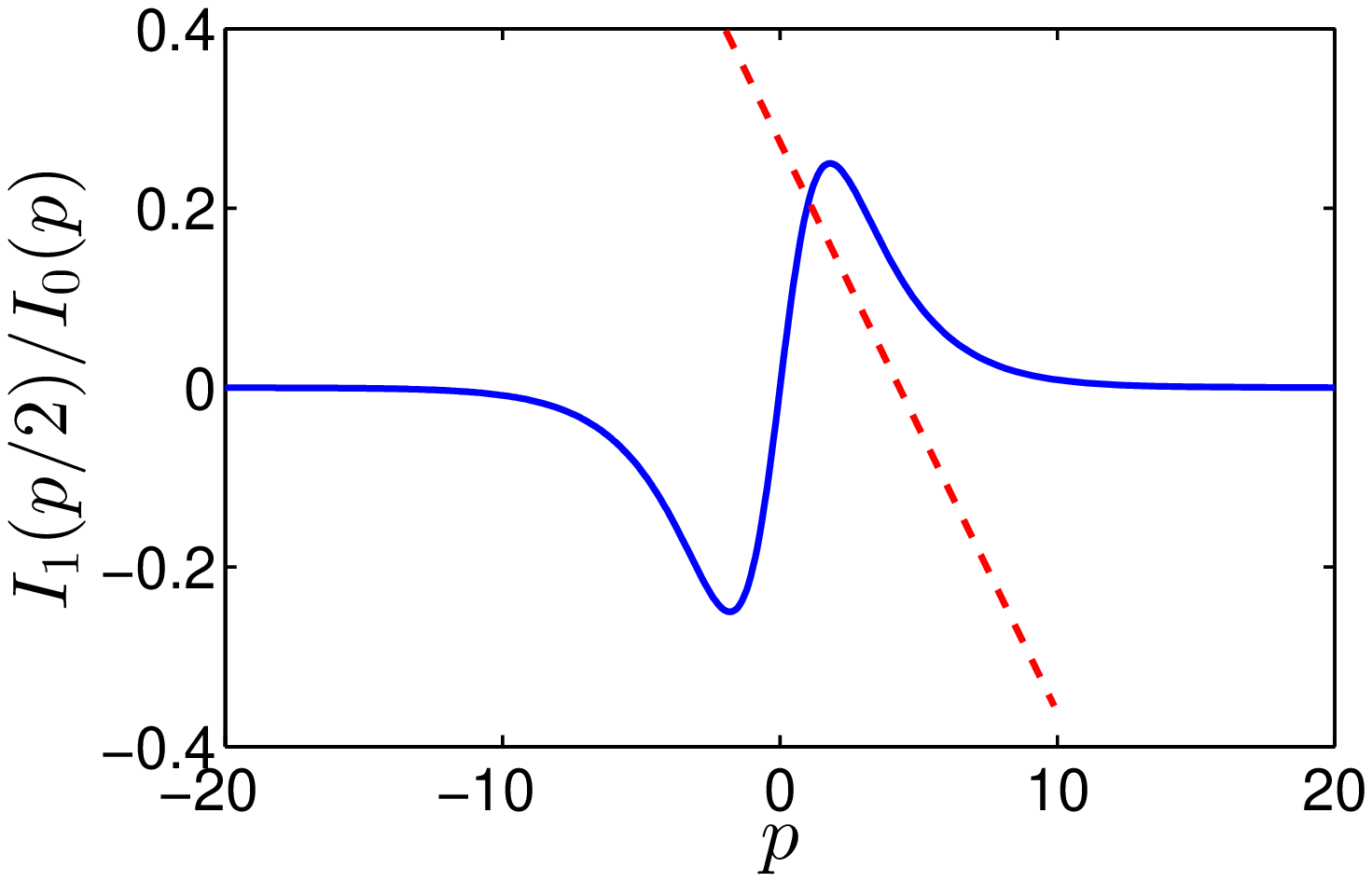}
\includegraphics[height=4cm]{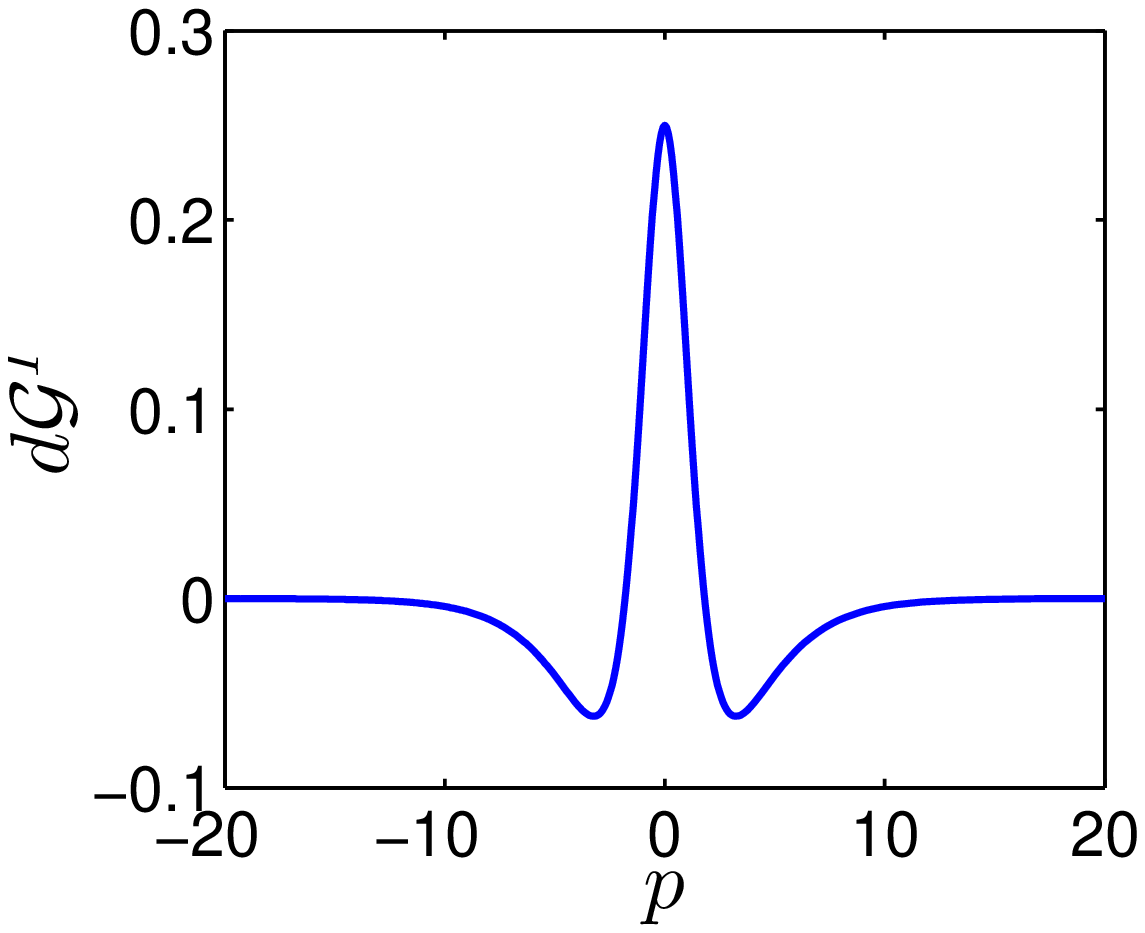}
\caption{Left: Intersection of the function $I_1(p/2)/I_0(p)$ (blue, solid) 
with $2a_p/b_p - 2p/b_p$ (red, dash) is the solution of $p_2$ defined by Equation (\ref{eqn:mapG}). Two or three intersections could 
be found, but when the slope $2/b_p$ is large enough the interaction is unique. Right: 
Function $d\mathcal{G}^I(p)$ has a maximum value of $0.25$ at $p=0$.} 
\label{fig:mapG}
\end{center}
\end{figure}

Existance of the unique fixed point of $\mathcal{G}$ allows us to use the iteration
$$ p^{(k+1)} = \mathcal{G}(p^{(k)})$$
and an arbitrary initial value $p^{(0)}$ to produce a convergent sequence $\{ p^{(k)} \}$ whose 
limit gives the solution of $p_2$ from which $T_1$ is obtained. The iteration is effcient because 
$\mathcal{G}$ is a map for one single variable. Other numerical methods such as the Newton's method
is also applicable. The remaining $T_k$ can then be uniquely determined using 
\begin{equation}
T_k = 2 \pi \cdot \frac{I_0(p_2/2)}{I_k(p_2/2)}, \quad k \ge 2.
\end{equation}
One can terminate at any integer $N_{\psi}$ to finally get an approximate solution
$$ \psi_N(\theta) = \sum_{k=0}^{N_{\psi}} T_k \cos( 2k (\theta_m - \theta)).$$
Sometime we might need the average angle of orientation $\< \theta \>$ corresponding to the distribution function $\psi(\theta)$.
It can be computed as 
\begin{eqnarray}
\< \theta \> & = & \int_0^{2 \pi} \psi_N(\theta) \theta d \theta  \nonumber \\
             & = & \pi + \sum_{k=1}^N T_k \int_0^{2 \pi} \cos(2k(\theta_m - \theta)) \theta d \theta  \nonumber \\
             & = & \pi \left( 1 - \sum_{k=1}^N \frac{T_k}{k} \sin(2k\theta_m) \right).
\end{eqnarray}

\begin{remark}
The existance of at most three intersections in the left chart of Fig. \ref{fig:mapG} indicates that the nonlinear map $\mathcal{F}$ 
defined by Equation (\ref{eqn:map_F}) has at most three fixed points.
\end{remark}

\begin{remark}
The time-dependent solution, if needed, can be obtained following combining the fixed point map $\mathcal{G}$ with the 
time-discrete, iterative variational scheme constructed for the standard Fokker-Planck equation \cite{JordanR1998a}.
\end{remark}



\bibliographystyle{siamplain}

\end{document}